\documentclass[aps,pra,twocolumn,superscriptaddress,floatfix,nofootinbib,showpacs,longbibliography,groupedaddress]{revtex4-2}

\usepackage[utf8]{inputenc}  
\usepackage[T1]{fontenc}     
\usepackage[british]{babel}  
\usepackage[sc,osf]{mathpazo}
\usepackage[scaled=0.86]{berasans}  
\usepackage[colorlinks=true, citecolor=blue, urlcolor=blue]{hyperref}  
\usepackage{graphicx} 
\usepackage[babel]{microtype}  
\usepackage{amsmath,amssymb,amsthm,bm,amsfonts,mathrsfs,bbm} 

\usepackage{xspace}  
\usepackage[table]{xcolor}
\usepackage{pgf,tikz}
\usepackage{xcolor}
\usepackage{multirow}
\usepackage{array}
\usepackage{bigstrut}
\usepackage{braket}
\usepackage{color}
\usepackage{natbib}
\usepackage{multirow}
\usepackage{mathtools}
\usepackage{float}
\usepackage[caption = false]{subfig}
\usepackage{xcolor,colortbl}
\usepackage{color}
\usepackage{booktabs}
\usepackage{scalerel}
\usepackage{mathtools}
\usepackage{comment}
\usepackage{amsmath}
\usepackage{ulem}

\newcommand\Mycomb[2][^n]{\prescript{#1\mkern-0.5mu}{}C_{#2}}

\newcommand{\be}{\begin{equation}}
\newcommand{\ee}{\end{equation}}
\newcommand{\ba}{\begin{eqnarray}}
\newcommand{\ea}{\end{eqnarray}}

\newcommand{\tr}{\operatorname{Tr}}

\newcommand{\etal}{{\it{et. al. }}}
\newcommand{\cf}{{\it{cf. }}}
\newcommand{\ie}{{\it{i.e.}}}

\newtheorem{theorem}{Theorem}

\newtheorem{lemma}{Lemma}

\def\>{\rangle}
\def\<{\langle}

\usepackage{centernot}
\usepackage{subfig}

\begin{document}

\title{Bounds on 
semi-device-independent quantum random number expansion capabilities}

\author{Vaisakh Mannalath}  \thanks{vaisakhmannalath@gmail.com}
\affiliation{Jaypee Institute of Information Technology, A-10, Sector 62, Noida UP 201309, India}

\author{Anirban Pathak}  \thanks{anirban.pathak@jiit.ac.in}
\affiliation{Jaypee Institute of Information Technology, A-10, Sector 62, Noida UP 201309, India}
	
\begin{abstract}
 The randomness expansion capabilities of semi-device-independent (SDI) prepare and measure protocols are analyzed under the sole assumption that the Hilbert state dimension is known. It's explicitly proved that the maximum certifiable entropy that can be obtained through this set of protocols is  $-\log_2\left[\frac{1}{2}\left(1+\frac{1}{\sqrt{3}}\right)\right]$ and the same is independent of the dimension witnesses used to certify the protocol. The minimum number of preparation and measurement settings required to achieve this entropy is also proven. An SDI protocol that generates the maximum output entropy with the least amount of input setting
 is provided. An analytical relationship between the entropy generated and the witness value is obtained. It's also established that certifiable entropy can be generated as soon as dimension witness crosses the classical bound, making the protocol noise-robust and useful in practical applications.
\end{abstract}


\maketitle

\section{Introduction}
 Randomness plays an important role in simulation algorithms \cite{Metropolis1949,Karp1991,10.1145/234313.234327}, cryptography \cite{6769090,1621063,PhysRevA.86.062308}, fundamental
sciences \cite{PhysicsPhysiqueFizika.1.195,Wheeler1978,Shadbolt2014} and much research has been devoted to the generation of random numbers \cite{RevModPhys.89.015004}. Deterministic algorithms can at best create `pseudo-random numbers' that mimic the statistics of `true' random numbers \cite{Ecuyer}. One needs access to unpredictable physical processes in order to generate truly random numbers \cite{RevModPhys.89.015004,Calude2015IndeterminismAR}. Quantum theory provides well-defined theoretical models which are inherently probabilistic and serve us with good entropy sources to extract randomness \cite{PhysRevA.87.062327}. Generating randomness from quantum systems is a matured field \cite{Ma2016}. There are now even commercially available quantum random number generators (QRNGs) \cite{QuantumR93:online,Jacak2021,Huang2021}. These devices are based on methods that are only applicable to their specific experimental setup and corresponding entropy estimates of the output randomness depend on a number of assumptions. Ultimately, these devices require a level of trust in the manufacturer which is not ideal for a number of reasons \cite{RevModPhys.89.015004}.

For the above-mentioned reasons, it is highly advantageous to have a setup that provides certifiable entropy while making minimal assumptions about its working. Device-independent QRNGs (DI-QRNGs) \cite{Colbeck2009,colbeck2011} provide a solution to this problem. By consuming input randomness and using non-locality of quantum theory it can, theoretically, certify the output randomness without characterizing the inner workings of the setup. There has also been numerous experimental demonstrations of this approach \cite{Pironio2010,PhysRevA.87.012336,Liu2018,Liu2021}. However,  protocols for DI-QRNG suffer from practical issues which make them hard to implement outside of a laboratory setup compared to one-way protocols commonly used in commercial devices.

A more practical approach to random number generation is provided by the so-called semi-device-independent QRNGs (SDI-QRNGs) \cite{PhysRevA.84.034301,PhysRevA.85.052308,PhysRevA.100.062338,PhysRevA.94.060301,PhysRevX.6.011020,PhysRevLett.126.210503}. Unlike, DI-QRNG, complete knowledge of a part of the setup used for random number generation is allowed in SDI-QRNGs. Even though this incurs a  weaker form of security compared to the DI counter part, it is much more practical. Realistically, there might be parts of the device that are more error-prone than others. The SDI approach lets you design protocols that can still generate certifiable randomness while leaving such parts uncharacterized \cite{PhysRevLett.114.150501,PhysRevApplied.7.054018,PhysRevApplied.15.034034,Pivoluska2021}. These protocols are also easier to implement since non-local sources are not required and is thus more consumer-friendly. Hence the entropy generation capabilities of the SDI protocols are of particular interest to cryptographers and others who use random numbers for various practical purposes.\\

In this work, we derive a general upper bound on the amount of entropy generated by a class of SDI protocols. Specifically, we consider prepare and measure protocols of two-dimensional systems and two-outcome measurements. Even though various protocols belonging to this class have been studied previously \cite{PhysRevA.84.034301,PhysRevA.85.052308}, their analysis has been restricted to some particular dimension witnesses which are used to distinguish quantum processes from classical processes. Our results, however, are independent of dimensional witnesses. We prove that the maximum amount of entropy which could be generated by any protocol of this class is equal to $-\log_2\left[\frac{1}{2}\left(1+\frac{1}{\sqrt{3}}\right)\right]$. Moreover, the minimum number of preparation and measurement settings to certify this much entropy is also proven. We give an explicit example of a unique protocol that matches these bounds, proving them to be tight.
Furthermore, we derive an analytical relationship between the witness values and the entropy generated with this protocol.\\

 The rest of the paper is organized as follows. In Section \ref{sec:SDI model}, we briefly describe the SDI model and state some definitions that we use in the subsequent sections. Section \ref{sec:results} contains results on the limits of output/input randomness. We report an explicit protocol matching these limits and its subsequent analysis in Section \ref{sec:protocol}. In Section \ref{sec:discussion}, we present a brief discussion along with some relevant open questions for further research.

 \section{Semi-Device Independent Model\label{sec:SDI model}}
 We first illustrate the general structure of the SDI-QRNG protocol that we have considered here. It involves two black boxes shielded from the outside world (see \autoref{blackbox}). One of the devices (boxes) is used for state preparation while the other one is used for the measurement.
 The preparation black box, $A$ has $\mathcal{X}$ settings, and the measurement black box, $B$ has $\mathcal{Y}$ settings; $\mathcal{X}, \mathcal{Y}\geq2$ . Depending on the randomly chosen setting among $\mathcal{X}$, $A$ outputs a quantum system $\rho_x$, $x\in[\mathcal{X}]$ (we use $[N]$ to denote a set of cardinality $N$), which will then be sent to the second black box $B$ for measurement. We assume that the state $\rho_x\in \mathbb{C}^2$; is a two-dimensional system. The measurement device takes $\rho_x$ as input and measures it in one of the randomly chosen settings $\mathcal{Y}$ and outputs $b\in\{0,1\}$. This forms one round of the prepare and measure protocol. We can repeat this procedure multiple times to get a probability distribution given by
 \begin{equation}
  p(b|x,y)=\tr\left(\rho_x M^b_y\right),   
 \end{equation}
 where $M^b_y$ is the measurement operator acting on $\rho_x$ with input parameter $y\in[\mathcal{Y}]$ and output  $b$.\\
 In order to identify whether the probability distributions truly have a quantum origin or not, dimension witnesses of the form
\begin{equation}
 W\equiv\sum_{x,y}w_{x,y}E_{x,y}
 \label{witnessdefinition}
\end{equation}
 are usually used, where $w_{x,y}$ are real coefficients and $$E_{x,y}=P(b=0|x,y).$$ Under such dimension witnesses, an SDI protocol does not demand any restriction on pre-shared classical correlations between the preparation and measurement devices \cite{PhysRevLett.105.230501}. Although we do assume that they don't share any quantum correlations. If we denote by $W_c$
 and $W_Q$ the classical and quantum upper-bounds of the witness value using two-dimensional systems, whenever,
 \begin{equation}
     W_c< W \leq W_Q
     \label{limitofw}
 \end{equation}
  we can be certain that the protocol has no classical description \cite{PhysRevLett.105.230501,PhysRevA.84.034301}. Hence the output $b$ of $B$ is truly probabilistic in nature and can be used to extract randomness \cite{10.1145/502090.502099,Carter1979}.

The entropy in the output $b$ can be quantified by the following min-entropy function \cite{5208530}
\begin{equation}
     H_\infty(B|\mathcal{X},\mathcal{Y})=-\log_2\left[\max_{b,x,y}p(b|x,y)\right].
     \label{entropy}
\end{equation}
This entropy is considered to be `certifiable' if the corresponding probability distribution satisfies the constraint \autoref{limitofw}.\\
 Since our witnesses defined by \autoref{witnessdefinition} are linear in probabilities, we just need to consider pure states for our analysis as any arbitrary mixed state can be written as a convex combination of pure states \cite{PhysRevLett.105.230501}. It has also been proven that POVMs can be depicted as
a convex combinations of projective measurements in the case of two-measurement outcomes \cite{masanes2005extremal,Tomamichel2013}. 
Furthermore, its known that projective measurements on two dimensional systems can be represented as antipodal unit vectors on the Bloch sphere. In general, the basis elements can be expressed as
\begin{equation}
 M^0_y=\frac{1}{2}(\mathbb{I}+\overrightarrow{t}_y\cdot\sigma) ~~~~~~~~~M^1_y=\frac{1}{2}(\mathbb{I}-\overrightarrow{t}_y\cdot\sigma), \label{measvector}
 \end{equation} where $\overrightarrow{t}_y$ is a unit vector on the Bloch sphere and $\sigma=(\sigma_x,\sigma_y,\sigma_z)$, the Pauli matrices. For preparations, it is  enough to consider pure states represented using unit vectors $\overrightarrow{s}_x$ as
\begin{equation}
 \rho_x=\frac{1}{2}\left(\mathbb{I}+\overrightarrow{s}_x\cdot\sigma\right). \label{statevector}
 \end{equation}
 Under this representation, probability distribution $p(b|x,y)$ can be expressed as
 \begin{equation}
p(b|x,y)=\tr\left(\rho_x M^b_y\right)=\frac{1}{2}\left(1+\overrightarrow{s}_x\cdot \overrightarrow{t}_y\right).
\end{equation}

We may now proceed to prove some general results related to the capabilities of SDI-QRNGs using the definitions and notations introduced in this section.
\begin{figure}
     \centering
     \includegraphics[width=0.9\linewidth]{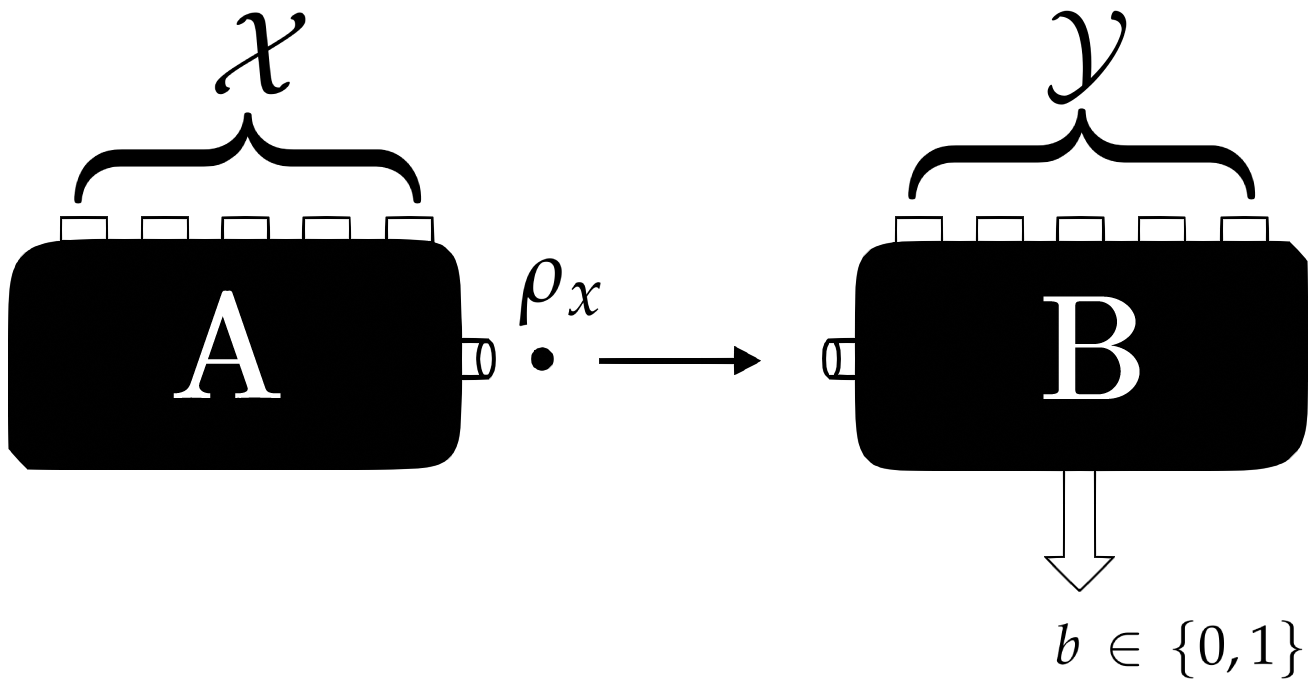}
     \caption{We consider the prepare and measure scenario of two-dimensional systems, $\rho_x$. Protocol features two black boxes $A$ and $B$, for preparation and measurement, respectively. $A$ has $\mathcal{X}$ input settings and $B$ has $\mathcal{Y}$ input settings. Based on the input, $A$ will output a state $\rho_x$ and $B$ will output $b\in\{0,1\}$ based on its input and the state sent by $A$.}
     \label{blackbox}
 \end{figure}
 \section{Results: Bounds on certifiable entropy}
 \label{sec:results}
 
 \begin{theorem}
  A prepare and measure protocol of two-dimensional systems and two-outcome measurements 
  can generate at most $-\log_2\left[\frac{1}{2}\left(1+\frac{1}{\sqrt{3}}\right)\right]$ bits of certifiable entropy.
 \label{theo1}
 \end{theorem}
 \begin{proof}
 
  Maximizing the entropy in \autoref{entropy} amounts to minimizing the quantity $\max_{b,x,y}p(b|x,y)$ over all prepare and measure protocols. We shall achieve this by first defining a lower bound for the quantity to be minimized and then deriving the lowest possible value for the lower bound.\\
 Consider the quantity $$p_{lb}=\max_x\frac{1}{\mathcal{Y}}\sum_{y}\max_b p(b|x,y)$$
 where the average is taken over all the measurement settings. Since mean over a set is lower than the maximum of a set, $p_{lb}$ forms a lower bound to $\max_{b,x,y} p(b|x,y)$.
 We may now derive an expression for $p_{lb}$.\\
 Let us represent the measurement basis for  $\mathcal{Y}$ measurement settings as $$\mathcal{T}_\mathcal{Y}=\{\{\overrightarrow{t}_1,-\overrightarrow{t}_1\},\{\overrightarrow{t}_2,-\overrightarrow{t}_2\},\cdots,\{\overrightarrow{t}_\mathcal{Y},-\overrightarrow{t}_\mathcal{Y}\}\}.$$
 Each of these measurement bases can be represented by a diameter of the Bloch sphere with the basis elements as its endpoints. For example the basis $\{\overrightarrow{t}_1,-\overrightarrow{t}_1\}$ represents a diameter with $\overrightarrow{t}_1$  and  $-\overrightarrow{t}_1$ as its endpoints. It is trivial to see that given any two non-perpendicular diameters of a sphere in ${\rm I\!R}^3$ the smallest angle between them would be less than or equal to $\pi/2$.  To further illustrate our point, let us consider the bases $\{\overrightarrow{t}_i,-\overrightarrow{t}_i\}$and$\{\overrightarrow{t}_j,-\overrightarrow{t}_j\}$. If the angle between the vectors $\overrightarrow{t}_i$ and $-\overrightarrow{t}_j$ is greater than $\pi/2$, then the angle between $\overrightarrow{t}_i$ and $\overrightarrow{t}_j$ will definitely be less than $\pi/2$; for any $i,j \in [\mathcal{Y}]$.\\
 Keeping the above arguments in mind, consider a particular case: $\mathcal{Y}=2$, with measurement bases as $\{\overrightarrow{t}_1,-\overrightarrow{t}_1\}$ and $\{\overrightarrow{t}_2,-\overrightarrow{t}_2\}$. If we consider the angle between $\overrightarrow{t}_1$ and $\overrightarrow{t}_2$ to be less than or equal to $\pi/2$, then the state $\rho_x$ with $\overrightarrow{s}_x=\frac{\overrightarrow{t}_1+\overrightarrow{t}_2}{|\overrightarrow{t}_1+\overrightarrow{t}_2|}$ maximizes $\frac{1}{\mathcal{Y}}\sum_{y}\max_b p(b|x,y)$, yielding $p_{lb}$. It is easy to see that $p_{lb}$ is minimum when $\overrightarrow{t}_1$ and $\overrightarrow{t}_2$ are perpendicular to each other. \\
 For now, let us assume that $0\leq\theta_{i,j}\leq\pi/2$, where $\theta_{i,j}$ is the angle between the measurement vectors $\overrightarrow{t}_i $ and $\overrightarrow{t}_j$ for $i,j \in [\mathcal{Y}]$. This is in general not true  for $\mathcal{Y}\geq3,$ but we will give an argument at the end of the proof as to why this assumption is valid enough to find out the minimum value of $p_{lb}$.\\
 Given this setup, consider $\rho_x$ with $\overrightarrow{s}_x=\frac{\overrightarrow{t}_1+\overrightarrow{t}_2+\cdots+\overrightarrow{t}_\mathcal{Y}}{|\overrightarrow{t}_1+\overrightarrow{t}_1+\cdots+\overrightarrow{t}_\mathcal{Y}|}$. Note that given the choice of measurement vectors $\{\overrightarrow{t}_1,\overrightarrow{t}_2,\cdots,\overrightarrow{t}_\mathcal{Y}\}$, this state maximizes $\frac{1}{\mathcal{Y}}\sum_{y}\max_b p(b|x,y)$ since it lies along the average direction of the measurement vectors. Simplification yields
\begin{equation}
p_{lb}=\frac{1}{2}\left(1+\frac{|\overrightarrow{t}_1+\overrightarrow{t}_2+\cdots+\overrightarrow{t}_\mathcal{Y}|}{\mathcal{Y}}\right).
\label{eqmod}
\end{equation}
We can represent \autoref{eqmod} as
 \begin{equation}
     p_{lb}=\frac{1}{2}\left(1+\frac{\sqrt{\mathcal{Y}+2(\cos{\theta_{1,2}}+\cdots+\cos{\theta_{\mathcal{Y}-1,\mathcal{Y}}}})}{\mathcal{Y}}\right).
     \label{plbang}
 \end{equation}
 \begin{lemma}
 For $\mathcal{Y}$  unit vectors that lie in an octant of a sphere in $\mathbb{R}^3$, the minimum of the sum of cosines of the angles formed between them is equal to $\frac{3}{2}\mu(\mu-1)+r\mu$ where $\mathcal{Y}= 3\mu+r$ for positive integers $\mu$ and $r\in\{0,1,2\}$.
 \label{sumofcos}
 \end{lemma}
 \begin{proof}
 Consider $\mathcal{Y}$ vectors $\{\overrightarrow{t}_1,\cdots,\overrightarrow{t}_\mathcal{Y}\}$. The sum to be minimized is
 $$
 \cos{\theta_{1,2}}+\cos{\theta_{1,3}}+\cdots+\cos{\theta_{\mathcal{Y}-1,\mathcal{Y}}}.
 $$
 We can rewrite it as
 $$
 \left(\cos{\theta_{1,2}}+\cos{\theta_{1,3}}+\cdots+\cos{\theta_{1,\mathcal{Y}}}\right)+\cdots+\cos{\theta_{\mathcal{Y}-1,\mathcal{Y}}}.
 $$
 The terms in the bracket is equal to
 $$
 \overrightarrow{t}_1\cdot\left(\overrightarrow{t}_2+\cdots+\overrightarrow{t}_\mathcal{Y}\right).
 $$
 Since every vector lies in the same octant, the dot product is minimized when $\overrightarrow{t}_1$ is along one of the axis. We can repeat the same process for every other vector until all of them lines up with one of the $3$ axes. We have three scenarios based on the value of $r\in\{0,1,2\}$; \\
\begin{itemize}
    \item $\mathcal{Y}=3\mu$: It is trivial to see that the sum is minimum when the vectors are equally distributed among the axes; $\mu$ vectors along each axis. The sum of cosines is equal to $3 \Mycomb[\mu]{2}$.
    \item $\mathcal{Y}=3\mu+1$: An additional vector along any one of the axes, say $x$ axis. The sum becomes $2 \Mycomb[\mu]{2}+\Mycomb[\mu+1]{2}$.
    \item $\mathcal{Y}=3\mu+2$: Consider the sum of vectors
    $$\overrightarrow{t}_1+\cdots+\overrightarrow{t}_{3\mu+1}.$$
    Since they have an arrangement dictated by the previous case, the vector $\overrightarrow{t}_{3\mu+2}$ should end up at $y$ or $z$ axis. The sum becomes $\Mycomb[\mu]{2}+2\Mycomb[\mu+1]{2}$.
\end{itemize}
Putting it all together, we have the sum as
$$
(3-r)\Mycomb[\mu]{2}+r\Mycomb[\mu+1]{2}.
$$
Simplifying it we obtain
$$
\frac{3}{2}\mu(\mu-1)+r\mu.
 $$
 \end{proof}
 Since our measurement vectors $\{\overrightarrow{t}_1,\overrightarrow{t}_2,\cdots,\overrightarrow{t}_\mathcal{Y}\}$ are atmost $\pi/2$ away from each other, we can consider them to lie in the same octant. Applying Lemma \ref{sumofcos}, \autoref{plbang} becomes,
\begin{equation}
p_{lb}=\frac{1}{2}\left(1+\frac{\sqrt{\mathcal{Y}+3\mu(\mu-1)+2r\mu}}{\mathcal{Y}}\right).
\label{plbsubs}
\end{equation}
Substituting $\mathcal{Y}=3\mu+r$ we obtain
 \begin{eqnarray}
   p_{lb}&=&\frac{1}{2}\left(1+\frac{\sqrt{3\mu^2+r(2\mu+1)}}{3\mu+r}\right).
 \end{eqnarray}
 Thus, for $r=0$, we have
 \begin{eqnarray}
 p_{lb}&=&\frac{1}{2}\left(1+\frac{\sqrt{3\mu^2}}{3\mu}\right)\nonumber\\
   &=&\frac{1}{2}\left(1+\frac{1}{\sqrt{3}}\right).
 \end{eqnarray}
 For $r\in\{1,2\}$, $p_{lb}>\frac{1}{2}\left(1+\frac{1}{\sqrt{3}}\right)$ and $p_{lb}\rightarrow\frac{1}{2}\left(1+\frac{1}{\sqrt{3}}\right)$ as $\mu\rightarrow \infty$.
 Thus in general,
 $$
 p_{lb}\geq\frac{1}{2}\left(1+\frac{1}{\sqrt{3}}\right).
 $$
 \end{proof}
Note that our assumption that $0\leq\theta_{i,j}\leq\pi/2$ for $i,j \in [\mathcal{Y}]$ is valid enough since the minimum value for $p_{lb}$ is obtained when the measurement vectors are along the $3$D axes. This implies that by induction, using $\mathcal{Y}=2$ as the initial step and the freedom to relabel any measurement basis, we can take any $\mathcal{Y}$ measurement vectors to lie in the same octant.
\begin{theorem}A prepare and measure protocol of two-dimensional systems needs at least $4$ preparation settings and $3$ measurement settings to generate the maximum amount of entropy. \label{theo2}
\end{theorem}
\begin{proof}
From \autoref{theo1}, the maximum entropy is generated when $\mathcal{Y}=3\mu$. For $\mu=1$, we have the minimum number of measurement settings, $\mathcal{Y}=3$. We will now try to minimise the number of preparation settings when $\mathcal{Y}=3$.\\
As for the number of preparation settings $\mathcal{X}$, note that it can't be $2$ since the states will be perfectly distinguishable; there is no entropy in the output. When $\mathcal{X}=3$ and $\mathcal{Y}=3$, a general dimension witness defined by \autoref{witnessdefinition} can be expressed as
\begin{eqnarray}
W&=&w_{1,1}E_{1,1}+w_{1,2}E_{1,2}+w_{1,3}E_{1,3}+w_{2,1}E_{2,1}+w_{2,2}E_{2,2}\nonumber\\
&&+w_{2,3}E_{2,3}+w_{3,1}E_{3,1}+w_{3,2}E_{3,2}+w_{3,3}E_{3,3},\label{witnessfor33}
\end{eqnarray}

From \autoref{theo1}, maximum entropy generation needs at least $3$ measurement settings. Since maximization is over the entire probability distribution, this holds for every preparation setting. Hence, none of the coefficients $w_{x,y}$ can be $0$ for a witness which acheives the maximum entropy. For example, suppose $w_{3,3}$ is $0$. This implies that the state $\rho_3$ depends only on the measurement bases $M_1$ and $M_2$; $\rho_3$ lies in the plane defined by $M_1$ and $M_2$. Subsequently, for a given witness value, $\max_{b,x,y}p(b|x,y)\geq\frac{1}{2}\left(1+\frac{1}{\sqrt{2}}\right)>\frac{1}{2}\left(1+\frac{1}{\sqrt{3}}\right)$.\\
Now that we have established that all coefficients in \autoref{witnessfor33} are non-zero, we can model it as an QRAC-like protocol where preparation states correspond to $3$ bit strings and measurement settings determine which bit to guess. Positive coefficients are mapped to bit $0$ and negative coefficients to bit $1$. For example, consider
\begin{eqnarray}
R_{3,3}&\equiv&E_{1,1}+E_{1,2}+E_{1,3}+E_{2,1}-E_{2,2}\nonumber\\
&&-E_{2,3}-E_{3,1}+E_{3,2}-E_{3,3}.
\label{witnessfor33example}
\end{eqnarray}
Based on our construction,  $R_{3,3}$ can be defined as the average success probability of an QRAC protocol where preparation states are represented as $x\in\{000,011,101\}$ and measurement settings dictate which one of the three bits to guess.

For any such  task we can construct a protocol based on the $2\rightarrow1$ QRAC (\cf \autoref{Bloch2}) whose average probability would be greater than $\frac{1}{2}(1+\frac{1}{\sqrt{3}})$ , implying the entropy generated will be lesser than $-\log_2\left[\frac{1}{2}(1+\frac{1}{\sqrt{3}})\right]$ (since $\max_{b,x,y}p(b|x,y)\geq \sum_{x,y}p(b=x_y|x,y)$: $x$ represents any of the $3$-bit strings and $x_y$ denotes the $y^{th}$ bit of that string).\\
The protocol is represented using \autoref{Bloch2}. A $3$ bit string can be written as $x$ where $x\in\{00x_3,01x_3,10x_3,11x_3\}$ and $x_3$ is the third bit (which could be different for different strings). The task is then straightforward; encode the strings on any three of the four possible states using $2\rightarrow1$ QRAC, where the encoding can be represented as
\\
\begin{align*}
\hspace{0.65cm}&\mbox{Encoding}\begin{cases}00\to\frac{1}{2}\left(\mathbb{I}+\frac{1}{\sqrt{2}}\sigma_x+\frac{1}{\sqrt{2}}\sigma_y\right)\\
01\to\frac{1}{2}\left(\mathbb{I}+\frac{1}{\sqrt{2}}\sigma_x-\frac{1}{\sqrt{2}}\sigma_y\right)\\
10\to\frac{1}{2}\left(\mathbb{I}-\frac{1}{\sqrt{2}}\sigma_x+\frac{1}{\sqrt{2}}\sigma_y\right)\\
11\to\frac{1}{2}\left(\mathbb{I}-\frac{1}{\sqrt{2}}\sigma_x-\frac{1}{\sqrt{2}}\sigma_y\right).
\end{cases}
\end{align*}
Decode the first two bits using the measurement bases given by
$$M_y\equiv\left\{\frac{1}{2}(\mathbb{I}+\overrightarrow{t}_y\cdot\sigma),\frac{1}{2}(\mathbb{I}-\overrightarrow{t}_y\cdot\sigma)\right\}$$
and their corresponding Bloch vectors
\begin{align*}
&\mbox{Decoding}\begin{cases}x_1\to \overrightarrow{t}_i\equiv(1,0,0)\\
x_2\to \overrightarrow{t}_i\equiv(0,1,0).
\end{cases}
\end{align*}
Decode the third bit using
$$
x_3\to \overrightarrow{t}_3\equiv\frac{1}{\sqrt{2}}(1,1,0).
$$
Given the choice of measurement bases and prepared states
the average probability is found to be at least,
$$
\frac{6\left(\frac{1}{2}(1+\frac{1}{\sqrt{2}})\right)+2}{9}\approx0.79125.
$$
Hence for such a protocol we have average probability greater than $\sim$ $0.79125$ which is greater than $\frac{1}{2}(1+\frac{1}{\sqrt{3}})\approx0.78867$.
\begin{figure}
    \centering
    \includegraphics[width=0.8\linewidth]{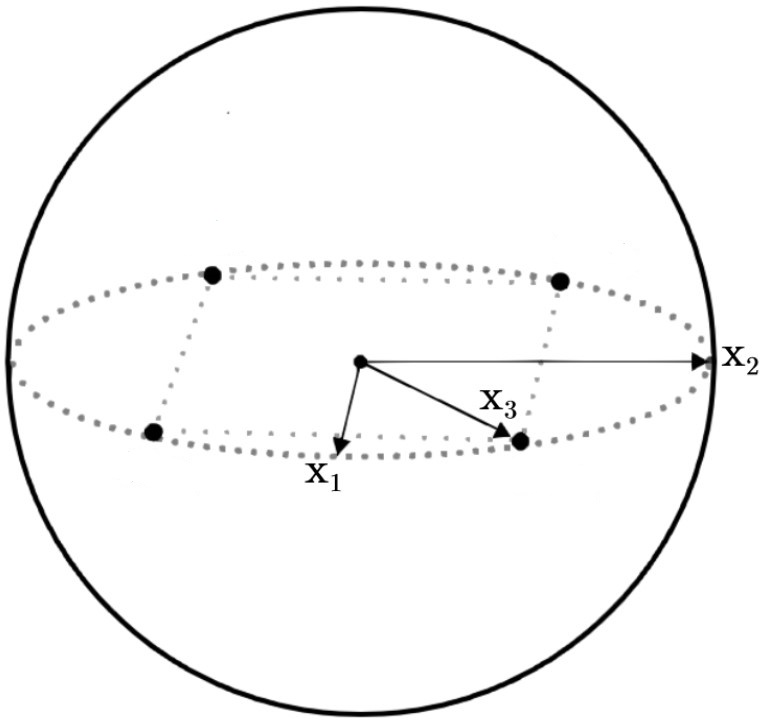}
    \caption{Bloch sphere diagram of a SDI protocol with $3$ preparation setting and $3$ measurement settings.  The arrows denote the 'up' direction of the measurement basis and the black dots indicate the encoded states for a particular choice of setting or equivalently, a string of bits. The $3$ strings are to be encoded in any of the four vertices. The encoded states of this protocol form a subset of the encoded states in the $2\rightarrow1$ QRAC, which forms a square on the equitorial plane of the Bloch sphere, denoted in this figure using dotted lines.}
    \label{Bloch2}
\end{figure}
\end{proof}

 \section{A specific protocol}
 \label{sec:protocol}
 An explicit protocol to achieve  $H_\infty=-\log_2[\frac{1}{2}(1+\frac{1}{\sqrt{3}})]\approx0.34249$ using $4$ preparation settings and $3$ measurement settings is given here.  It corresponds to an QRAC-like protocol in which the preparation party $A$ encodes one of the strings $x\in\{000,011,101,110\}$ to a single qubit and the measurement party $B$ attempts to decode $x_y\in\{x_1,x_2,x_3\}$
by suitable measurements. Using arguments similar to those used for proving  \autoref{theo2}, this protocol can be proven to be unique up to some relabelling of the measurement bases. A suitable dimension witness is provided by
\begin{multline}
 R_{4,3}\equiv E_{000,1}+E_{000,2}+E_{000,3}+E_{011,1}-\\
 E_{011,2}-E_{011,3}-E_{101,1}+E_{101,2}-\\
 E_{101,3}-E_{110,1}-E_{110,2}+E_{110,3},
 \label{witness}
\end{multline}
and whenever
 \begin{equation}
     3< R_{4,3}\leq2\sqrt{3},
     \label{witnessbounds}
 \end{equation}
 the protocol has no classical description. \autoref{witnessbounds} was derived from the results provided in \cite{PhysRevA.104.012420}. They have treated the protocol as a generalized version of the $2\rightarrow1$ QRAC protocol where $B$ attempts to decode, in addition to the bits encoded by $A$, the parity of the bits as well. The average success probability of this particular protocol have previously found applications in the SDI security of QKD protocols \cite{PhysRevA.84.010302}. The corresponding dimension witness has also been applied previously in self testing of POVMs \cite{doi:10.1126/sciadv.aaw6664,PhysRevA.100.030301} and in reduction of symmetric dimension witnesses \cite{PhysRevA.90.022322}. \\
In order to achieve the maximal quantum value $2\sqrt{3}$, we may
encode the bits using the states given by

\begin{align*} \label{32enc}
&\mbox{Encoding}\begin{cases}000 \mapsto\frac{1}{2}\left(\mathbb{I}+\frac{1}{\sqrt{3}}\sigma_x+\frac{1}{\sqrt{3}}\sigma_y+\frac{1}{\sqrt{3}}\sigma_z\right)\\
011 \mapsto\frac{1}{2}\left(\mathbb{I}+\frac{1}{\sqrt{3}}\sigma_x-\frac{1}{\sqrt{3}}\sigma_y-\frac{1}{\sqrt{3}}\sigma_z\right)\\
101 \mapsto\frac{1}{2}\left(\mathbb{I}-\frac{1}{\sqrt{3}}\sigma_x+\frac{1}{\sqrt{3}}\sigma_y-\frac{1}{\sqrt{3}}\sigma_z\right)\\
110 \mapsto\frac{1}{2}\left(\mathbb{I}-\frac{1}{\sqrt{3}}\sigma_x-\frac{1}{\sqrt{3}}\sigma_y+\frac{1}{\sqrt{3}}\sigma_z\right).\end{cases}\nonumber
\end{align*}
and decode the bits using the measurement bases given by
$$M_y\equiv\left\{\frac{1}{2}(\mathbb{I}+\overrightarrow{t}_y\cdot\sigma),\frac{1}{2}(\mathbb{I}-\overrightarrow{t}_y\cdot\sigma)\right\}$$
with the corresponding Bloch vectors
\begin{align*}
&\mbox{Decoding}\begin{cases} x_1\to \overrightarrow{t}_1\equiv(1,0,0)\\
x_2\to \overrightarrow{t}_2\equiv(0,1,0)\\
x_3\to \overrightarrow{t}_3\equiv(0,0,1).\end{cases}
\end{align*}

Note that a general two-dimensional witness for $4$ preparation settings and $3$ measurement settings may not be able to produce the maximum amount of randomness. For example, consider the well known dimension witness $I_4$ \cite{PhysRevLett.105.230501,Hendrych2012}, defined as
$$
\resizebox{\hsize}{!}{$I_{4} \equiv E_{1,1}+E_{1,2}+E_{1,3}+E_{2,1}+E_{2,2}-E_{2,3}+E_{3,1}-E_{3,2}-E_{4,1}$.}
$$
Since the choice of the $4^{th}$ state solely depends on the $1^{st}$ measurement basis, one can always take $E_{4,1}$ to be $0$. This implies that $p(b=1|4,1)=1$; no entropy is generated in this case. The choice of dimension witness is special in that regard and warrants further analysis. We will now derive an analytical bound on the min-entropy based on the value of the dimension witness. The analysis and methods used is similar to what have been done in \cite{PhysRevA.91.032305,Li2015}.\\
Using \autoref{measvector} and \autoref{statevector}, \autoref{witness} can be written as
 \begin{eqnarray}
 R_{4,3}&\equiv&E_{000,1}+E_{000,2}+E_{000,3}+E_{011,1}-E_{011,2}-E_{011,3}-\nonumber\\
&&E_{101,1}+E_{101,2}-E_{101,3}-E_{110,1}-E_{110,2}+E_{110,3}\nonumber\\
&=&\resizebox{.83\hsize}{!}{$\operatorname{\tr}\left[\rho_{000}\left(M_{1}^0+M_{2}^0+M_{3}^0\right)\right]+\operatorname{\tr}\left[\rho_{011}\left(M_{1}^0-M_{2}^0-M_{3}^0\right)\right]$}+\nonumber\\
&&\resizebox{.83\hsize}{!}{$\operatorname{\tr}\left[\rho_{101}\left(-M_{1}^0+M_{2}^0-M_{3}^0\right)\right]+\operatorname{\tr}\left[\rho_{110}\left(-M_{1}^0-M_{2}^0+M_{3}^0\right)\right]$}\nonumber\\
&=&\frac{1}{2}\biggl(\overrightarrow{s}_{000}\cdot\left(\overrightarrow{t}_1+\overrightarrow{t}_2+\overrightarrow{t}_3\right)\hspace{0.7em}+\nonumber\\&&\hspace{1.45em}\overrightarrow{s}_{011}\cdot\left(\overrightarrow{t}_1-\overrightarrow{t}_2-\overrightarrow{t}_3\right)\hspace{0.8em}+\nonumber\\
&&
\hspace{1.45em}\overrightarrow{s}_{101}\cdot\left(-\overrightarrow{t}_1+\overrightarrow{t}_2-\overrightarrow{t}_3\right)+\nonumber\\&&\hspace{1.45em}\overrightarrow{s}_{110}\cdot\left(-\overrightarrow{t}_1-\overrightarrow{t}_2+\overrightarrow{t}_3\right) \biggl) \nonumber\\
&\leq&\frac{1}{2}\biggl(\left|\overrightarrow{t}_1+\overrightarrow{t}_2+\overrightarrow{t_{3}  }\right|+\left|\overrightarrow{t}_1-\overrightarrow{t}_2-\overrightarrow{t}_3\right|+\nonumber\\
&&\hspace{1.65em}\left|\overrightarrow{t}_1-\overrightarrow{t}_2+\overrightarrow{t}_3\right|+\left|\overrightarrow{t}_1+\overrightarrow{t}_2-\overrightarrow{t}_3\right|\biggl)\nonumber\\
&\leq& \frac{1}{2}\biggl(\sqrt{3+2\left(\cos \left(\theta_{1,2}\right)+\cos \left(\theta_{1,3}\right)+\cos \left(\theta_{2,3}\right)\right)}+\nonumber\\
&&\hspace{1.67em}\sqrt{3+2\left(\cos \left(\theta_{1,2}\right)-\cos \left(\theta_{1,3}\right)+\cos \left(\theta_{2,3}\right)\right)}+\nonumber\\
&&\hspace{1.67em}\sqrt{3+2\left(\cos \left(\theta_{1,2}\right)+\cos \left(\theta_{1,3}\right)-\cos \left(\theta_{2,3}\right)\right)}+\nonumber\\
&&\hspace{1.67em}\sqrt{3+2\left(\cos \left(\theta_{1,2}\right)-\cos \left(\theta_{1,3}\right)-\cos \left(\theta_{2,3}\right)\right)}\biggl),\nonumber\\&&\label{Wlong}
\end{eqnarray}
 where the first inequality follows from $|\overrightarrow{s}_x|<1$. Using \autoref{plbang} we can write $p_{lb}$ for our example as
\begin{equation}
    p_{lb}=\frac{1}{2}\resizebox{0.86\hsize}{!}{$ \left(1+\frac{\sqrt{3+2\left(\cos \left(\theta_{1,2}\right)+\cos \left(\theta_{1,3}\right)+\cos \left(\theta_{2,3}\right)\right)}}{3}\right)$}.
     \label{plbfor3}
\end{equation}
 
 In order to obtain a bounded value of $R_{4,3}$ as a function of $p_{lb}$ we will use the extreme value problem of multi-variable function. Changing variables as
 $$
 \mathcal{P}=\frac{\left(3\left(2p_{lb}-1\right)\right)^2-3}{2}~~~~~a=\cos{\theta_{1,3}}~~~~~~b=\cos{\theta_{2,3}}
 $$
 and applying it to \autoref{Wlong}, we obtain
\begin{multline}
    R_{4,3}\leq\max_{\{(q,r)\}}\biggl\{\frac{1}{2}\biggl(\sqrt{3+2\mathcal{P}}+\sqrt{3+2\left(2q-\mathcal{P}\right)}+\\
    \sqrt{3+2\left(2r-\mathcal{P}\right)}+
    \sqrt{3+2\left(\mathcal{P}-\left(2q+2r\right)\right)}\biggl)\biggl\},
\end{multline}
where $\left(q,r\right)$ is  one of the real roots of equation set with variables $\left(a,b\right)$ given by,
\begin{eqnarray}
\frac{1}{2} \left(\frac{2}{\sqrt{2 (2 a-\mathcal{P})+3}}-\frac{2}{\sqrt{2 (\mathcal{P}-2 a-2 b)+3}}\right)=0,\nonumber\\
\frac{1}{2} \left(\frac{2}{\sqrt{2 (2 b-\mathcal{P})+3}}-\frac{2}{\sqrt{2 (\mathcal{P}-2 b-2 a)+3}}\right)=0.\nonumber\\
\label{dereq0}
\end{eqnarray}
The equation set provided above is obtained by taking the derivatives of \autoref{Wlong} with respect to $a$ and $ b$.
 It turns out that the solutions of \autoref{dereq0} should satisfy the condition
\begin{equation}
 a=b=\mathcal{P}/3
 \implies \cos{\theta_{1,2}}=\cos{\theta_{1,3}}=\cos{\theta_{2,3}}.
 \label{allangeq}
\end{equation}
\begin{figure}
\centering
\includegraphics[width=0.8\linewidth]{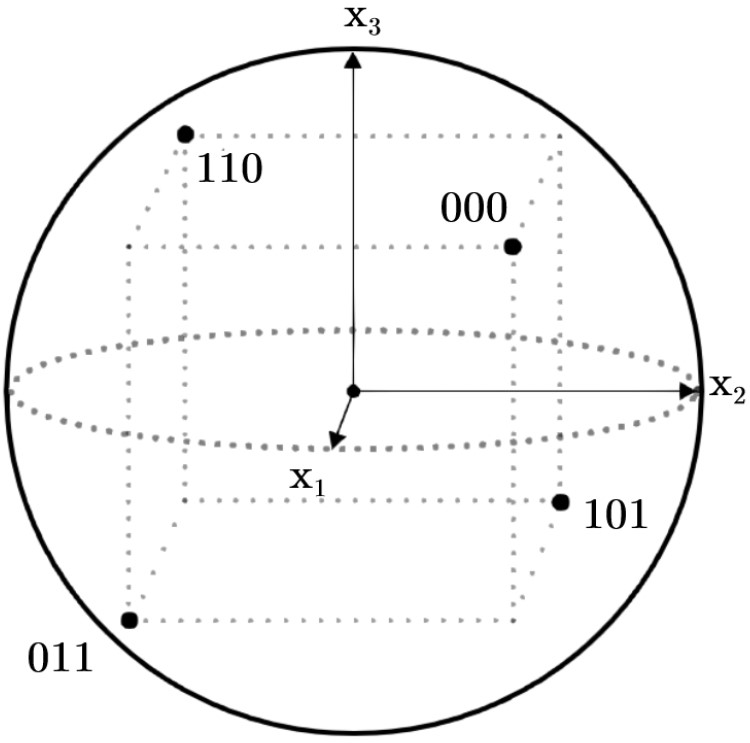}
\caption{Bloch sphere diagram of a SDI protocol with $4$ preparation setting and $3$ measurement settings. The arrows denote the 'up' direction of the measurement basis and the black dots indicate the encoded states for a particular choice of string/setting. The encoded states form a tetrahedron inside the Bloch sphere. They also form a subset of the encoded states in the $3\rightarrow1$ QRAC, which forms a cube, denoted in this figure using dotted lines.}
\label{Bloch3}
\end{figure}

Since $p_{lb}$ is defined as $\max_x\frac{1}{\mathcal{Y}}\sum_{y}\max_b p(b|x,y)$, by \autoref{allangeq} all the terms in the summation are equal. This means that when $R_{4,3}$ is maximized, $p_{lb}$ is equivalent to  $\max_{b,x,y}p(b|x,y)$. Hence the maximization will yield a tight upper bound on the randomness generated by this protocol. Also, \autoref{allangeq} reduces our problem to a single variable one. \\
Substituting \autoref{allangeq} in \autoref{Wlong} and \autoref{plbfor3}, we get
 \begin{eqnarray}
R_{4,3}&\leq&\frac{1}{2} \left(3 \sqrt{3-2 a}+\sqrt{6 a+3}\right),\nonumber\\
p_{lb}&=&\frac{1}{2} \left(\frac{1}{3} \sqrt{6 a+3}+1\right). 
\label{boundforplb}
\end{eqnarray}
 Solving \autoref{boundforplb} we obtain
 \begin{equation}
     p_{lb}\leq\frac{1}{12} \left(R_{4,3}+6 +\sqrt{3(12-R_{4,3}^2)}\right).
     \label{plbwitness}
 \end{equation}
\begin{figure}
    \centering
    \includegraphics[width=\linewidth]{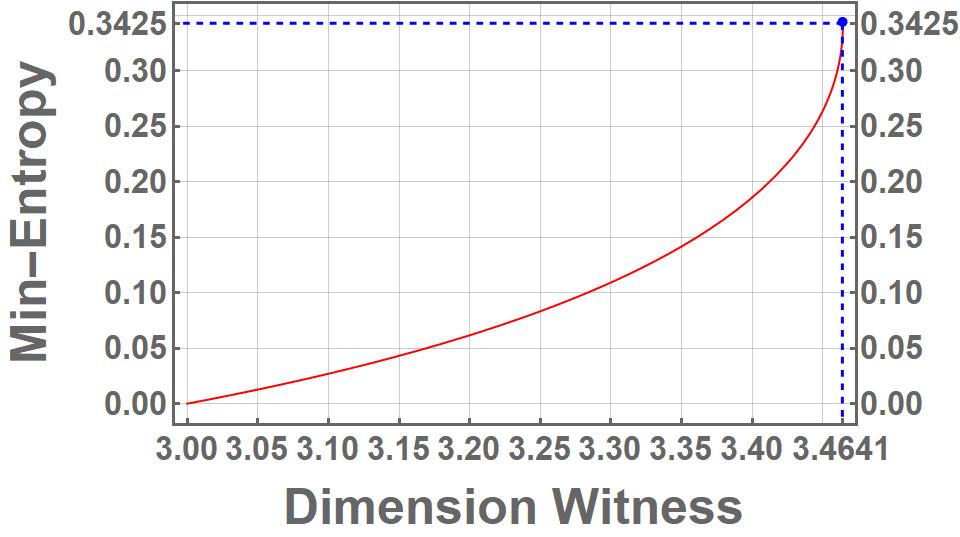}
    \caption{(Color online) Relationship between dimension witness value and upper bound on the entropy created by the protocol discussed in Section \ref{sec:protocol}.}
    \label{witvsentr}
\end{figure}
This forms a min-entropy bound for the particular protocol as shown in \autoref{witvsentr}. Since the choice of angles is unique when $R_{4,3}=2\sqrt{3}$ \ie,
$$
\theta_{1,2}=\theta_{1,3}=\theta_{2,3}=\pi/2,
$$
it yields the maximum amount of  certifiable randomness; $H_\infty\approx0.34249$. Also  note  that  since $p_{lb}$ forms an upper bound to the average success probability of the protocol, \autoref{plbwitness} implies that certifiable randomness can be generated as soon as one violates the classical bound on witness. This is particularly relevant in practical setups, which might not be able to achieve the maximum possible quantum violation. The protocols is thus noise-robust, and  has immediate applications in practical SDI-QRNG setups.\\
Since we assume that devices are shielded from the outside world, the randomness used to choose the input settings in each round can be used for other purposes. Hence the total output randomness from each round is more than what is being used to start the process. In order to increase randomness expansion even further, one can consider using a fixed subset of input setting for randomness generation for most rounds and a randomly chosen input setting for rest of the rounds \cite{Mironowicz2016,PhysRevA.91.032305,Bancal2014}. If the number of rounds is large enough, one can use the subset of rounds wherein the input setting where randomly chosen in order to estimate the witness value \cite{Pironio2010}. 
\section{Discussions and outlook}
\label{sec:discussion}
A tight bound on the entropy generation rate is derived for SDI prepare and measure protocols for two-dimensional systems and two-outcome measurements solely from geometrical arguments. The maximum entropy generated from such a class of protocols is found to be equal to $-\log_2\left[\frac{1}{2}\left(1+\frac{1}{\sqrt{3}}\right)\right]$. Here it will be apt to note the results of a previous work \cite{PhysRevA.99.052338} which suggests an  upper bound on the certifiable randomness from a quantum black box as $-\log_2\left[\min\{l,k+1\}\right]$, where $l$ is the number of outputs for a measurement ($2$ in our case) and $k$ is the number of preparation settings. For the particular class of protocols that we are considering, this result forms a trivial bound of $1$ bit of certifiable entropy. Our results are much more strict in that regard. It was also conjectured in \cite{PhysRevA.85.052308} that $3\rightarrow1$ QRAC generates the maximum amount of randomness among $n\rightarrow1$ QRAC protocols. We have proved that this is indeed the case. Our results are more general than QRAC protocols and also independent of any dimension witness. We have also provided an explicit protocol generating the maximum amount of entropy while having the least amount of input settings. The protocol generates as much entropy as $3\rightarrow1$ QRAC protocol does, however it requires lesser input settings. Note that even though the protocol generates maximum entropy when $W=W_Q$, it still remains an open question if one can extract more randomness than what is given in \autoref{witvsentr} when $W<W_Q$. Since \autoref{plbwitness} is tied to a specific dimension witness \ie, $R_{4,3}$, it would be worthwhile to investigate whether the methods by Wang \etal \cite{PhysRevA.92.052321} would be able to extract more randomness when $3<R_{4,3}<2\sqrt{3}$. Inspired by the DI approach used in \cite{Bancal2014,NietoSilleras2014}, they used the full observed statistics to certify randomness rather than restricting to a particular inequality.\\
 The results reported here are not only new and useful, it also opens up possibilities for a set of interesting investigations. An immediate generalisation of the presented work would be to consider the limits on entropy generation for $l$-outcome measurements on $d$-dimensional systems; $l,d>2$. It would also be interesting to investigate the randomness expansion capabilities of such protocols with partially free random sources as the input seed \cite{Colbeck2012,PhysRevA.92.022331}. Given the advantage of our protocol over the $3\rightarrow1$ QRAC with perfect random sources, it would be interesting to see a comparison with partially free sources \cite{PhysRevA.94.032318}. Another possible avenue for research would be to consider the randomness generation ability for multiple users as discussed in \cite{Wang2021}.

\section*{Acknowledgment}
Authors acknowledge the support from the QUEST scheme of Interdisciplinary Cyber Physical Systems (ICPS) program of the Department of Science and Technology (DST), India (Grant No.: DST/ICPS/QuST/Theme-1/2019/14 (Q80)). They also thank Manik Banik, Ram Krishna Patra, Sandeep Mishra, R Srikanth, H. S. Karthik, H. Akshata Shenoy and Kishore Thapliyal for their
interest and feedback on the work.

\bibliographystyle{apsrev4-2}
\bibliography{article-test.bib}

\end{document}